\documentclass[11pt]{llncs}

\usepackage{amssymb}

\usepackage{amsmath}

\def\bydef{\stackrel{{\rm def}}{=}}
\usepackage{algorithm}
\usepackage{algorithmic}
\usepackage{url}
\usepackage{xcolor} 

\newlength{\protowidth}
\newcommand{\pprotocol}[5]{
{\begin{figure}[#4]
\begin{center}
\setlength{\protowidth}{\textwidth}
\addtolength{\protowidth}{-2\intextsep}

\fbox{
        \small
        \hbox{\quad
        \begin{minipage}{\protowidth}
    \begin{center}
    {\bf #1}
    \end{center}
        #5
        \end{minipage}
        \quad}
        }
        \caption{\label{#3} #2}
\end{center}
\vspace{-4ex}
\end{figure}
} }

\newcommand{\protocol}[4]{
\pprotocol{#1}{#2}{#3}{tbh!}{#4} }
\newcommand{\signed}[2]{
	\langle #1\rangle_{#2} }

\def\bool{\{0,1\}}

\newcommand{\true}{\mathsf{true}}

\newcommand{\aba}{\mathsf{BA}}

\newcommand{\echo}{\mathsf{echo}}
\newcommand{\ready}{\mathsf{ready}}
\newcommand{\BB}{\Pi^{t_s}_{\sf BB}}

\newcommand{\ACS}{\Pi^{t_a, t_s}_{\sf ACS}}
\newcommand{\SMR}{\Pi^{t_a,t_s}_{\sf SMR}}
\newcommand{\tx}{{\sf tx}}

\newcommand{\verify}{\mathsf{Vrfy}}
\def\vrfy{\verify}
\newcommand{\buf}{\mathsf{buf}}

\newcommand{\BLA}{\Pi_{\mathsf{BLA}}}

\def\blocks{{\sf Blocks}}
\newcommand{\prp}{\Pi^{P^*}_{\mathsf{Propose}}}
\newcommand{\GC}{\Pi^{k}_{\mathsf{GC}}}

\def\epochs{{\sf Epochs}}
\newcommand{\block}[2]{\mathsf{Blocks}_{#1}[#2]}
\newcommand{\epoch}[2]{\mathsf{Epochs}_{#1}[#2]}

\newcommand{\sblock}{\mathsf{BlockSet}}
\newcommand{\leader}[1]{\mathsf{Leader}_{#1}}

\usepackage{cite}

\usepackage{authblk}

\newcommand{\ignore}[1]{}

\begin{document}


\title{{\bf Network-Agnostic State Machine Replication}}
\author{Erica Blum\inst{1} \and Jonathan Katz\inst{2}
\and Julian Loss\inst{1}}

\institute{University of Maryland
\and George Mason University}
\date{}      



\maketitle
\begin{abstract}
	We study the problem of \emph{state machine replication} (SMR)---the underlying problem addressed by blockchain protocols---in the presence of a malicious adversary who can corrupt some fraction of the parties running the protocol.
	Existing protocols for this task assume either a
	\emph{synchronous} network (where all
	messages are delivered within some known
	time~$\Delta$) or an
	\emph{asynchronous} network (where messages can be delayed arbitrarily). Although protocols for the latter case give seemingly stronger guarantees, this is not the case since they (inherently) tolerate a lower fraction of corrupted parties.
	
	\smallskip
	We design an SMR protocol that is \emph{network-agnostic} in the following sense: if it is run in a synchronous network, it tolerates $t_s$ corrupted parties; if the network happens to be asynchronous it is resilient to $t_a \leq t_s$ faults.
	Our protocol achieves optimal tradeoffs between $t_s$ and~$t_a$.
\end{abstract}

\setcounter{footnote}{0}
\section{Introduction}
\emph{State machine replication} (SMR) is a fundamental problem in distributed computing~\cite{Lamport78b,Paxos,Schneider90} that
can be viewed as a generalization of
\emph{Byzantine agreement} (BA)~\cite{PSL80,LSP82}.
Roughly speaking, a BA protocol allows a set of $n$~parties to agree on a value \emph{once},
whereas SMR allows those parties to
agree on an infinitely long \emph{sequence} of values with the additional guarantee that values input to honest parties are eventually included in the sequence.
(See Section~\ref{section:Definitions} for formal definitions. Note that SMR is not obtained by simply repeating a BA protocol multiple times;
see further discussion in Section~\ref{sec:related_work}.)
The desired properties should hold
even in the presence of some fraction of corrupted parties who may
behave arbitrarily.
SMR protocols are deployed in real-world distributed data centers, and the problem
has received renewed attention in the context of \emph{blockchain protocols} used for cryptocurrencies and other applications.

Existing SMR protocols assume either a \emph{synchronous network}, where all messages are delivered within some publicly known time bound~$\Delta$, or an \emph{asynchronous network}, where messages can be delayed arbitrarily.
Although it may appear that protocols designed for the latter setting are strictly more secure, this is not the case because they also (inherently) tolerate a lower fraction of corrupted parties.
Specifically, assuming a public-key infrastructure (PKI) is available to the parties, SMR protocols tolerating up to $t_s < n/2$ adversarial corruptions are possible in a synchronous network, but in an asynchronous network SMR is achievable only for $t_a < n/3$ faults
(see~\cite{DBLP:journals/cj/CorreiaNV06}).

We study here so-called \emph{network-agnostic} SMR protocols that offer meaningful guarantees regardless of the network in which they are run.
That is,
fix thresholds $t_a, t_s$ with $0 \leq t_a < n/3$ and $t_a \leq t_s < n/2$.
We seek to answer the following question: Assuming a PKI, is it possible to
have an SMR protocol that tolerates
(1)~$t_s$ (adaptive) corruptions if the network is synchronous and (2)~$t_a$ (adaptive) corruptions even if the network
is asynchronous? We show that the answer is positive iff $t_a + 2t_s < n$.

Our work is directly inspired by recent results of Blum et al.~\cite{BKL19}, who study the same problem but for the simpler case of Byzantine agreement. We match their bounds on~$t_a, t_s$ and, as in their work, show that these bounds are optimal in our setting.\footnote{It is not clear that SMR implies BA in the network-agnostic setting when $t_a + 2t_s \geq n$. Thus, impossibility of SMR when $t_a + 2t_s \geq n$ does not follow from the impossibility result for BA shown by Blum et al.~\cite{BKL19}.}
While the high-level structure of our SMR protocol resembles the high-level structure of their BA protocol, in constructing our protocol we need to address several technical challenges (mainly due to the stronger liveness property required for SMR; see the following section) that do not arise in their work.

\subsection{Related Work} \label{sec:related_work}

There is extensive prior work on designing both Byzantine agreement and SMR/blockchain protocols; we do not provide an exhaustive survey, but instead focus only on the most relevant prior work. 

As argued by Miller et al.~\cite{CCS:MXCSS16},
many well-known SMR protocols that tolerate malicious faults (e.g., \cite{PBFT,Zyzzyva}) require at least partial synchrony in order to achieve liveness. Their HoneyBadger protocol~\cite{CCS:MXCSS16} was designed
specifically for asynchronous networks, but can only handle $t<n/3$ faults even if run in a synchronous network. Blockchain protocols are typically analyzed assuming synchrony~\cite{EC:GarKiaLeo15,EC:PasSeeShe17}; Nakamoto consensus, in particular, assumes that messages will be delivered much faster than the time required to solve proof-of-work puzzles.

We emphasize that SMR is \emph{not} realized by simply repeating a (multi-valued) BA protocol multiple times. In particular, the validity property of BA only guarantees that if a value is input by all honest parties then that value will be output by all honest parties.
In the context of SMR the parties each hold multiple inputs in a local buffer (where those inputs may arrive at arbitrary times), and there is no way to ensure that all honest parties will select the same value as input to some execution of an underlying BA protocol.
Although generic techniques for compiling a BA protocol into an SMR protocol are known~\cite{DBLP:journals/cj/CorreiaNV06}, those compilers are not network-agnostic and so do not suffice to solve our problem.

Our work focuses on protocols being run in a network that may be either synchronous or fully asynchronous.
Other work looking at similar problems includes
that of Malkhi et al.~\cite{arxiv:MNR}, who consider
networks that may be either
synchronous or \emph{partially} synchronous;
Liu et al.~\cite{DBLP:conf/osdi/LiuVCQV16}, who design a protocol
that tolerates a minority of malicious faults in a
synchronous network, and a minority of \emph{fail-stop} faults
in an asynchronous network; and Guo et al.~\cite{GPS19} and Abraham et al.~\cite{AMNRY19}, who
consider temporary disconnections between two synchronous network components.

A slightly different line of work~\cite{PS17,EC:PasShi18,EPRINT:LosMor18,LLMMT19} looks at designing protocols with good \emph{responsiveness}. Roughly speaking, such protocols still require the network to be synchronous, but terminate more quickly if the actual message-delivery time is lower than the known upper bound~$\Delta$. Kursawe~\cite{Kursawe02} designed a protocol for an asynchronous network that terminates more quickly if the network is synchronous, but does not tolerate more faults in the latter case.
Finally, other work~\cite{DISC:FitNie09,PKC:DGKN09,PODC:BeeHirNie10,DBLP:journals/tit/PatraR18}
considers a model where synchrony is available
for some (known) limited period of time, but the network is asynchronous afterward.




\subsection{Paper Organization}

We define our model in
Section~\ref{section:Model}, before giving definitions for the various tasks we consider in Section~\ref{section:Definitions}.
In Section~\ref{sec:ACS-all} we describe a network-agnostic protocol for the asynchronous common subset (ACS) problem. The ACS protocol is used as a subprotocol of our main result, a network-agnostic SMR protocol, that is described and analyzed in Section~\ref{sec:SMR}.
In Section~\ref{sec:lower-bound} we prove a lower bound showing that the thresholds we achieve are tight for network-agnostic SMR protocols. As discussed, Blum et al.~\cite{BKL19} show an analogous result for BA that does not directly apply to our setting.

\section{Model}\label{section:Model}
\textbf{Setup assumptions and notation.}
We consider a network of $n$ parties $P_1,\ldots, P_n$ who communicate
over point-to-point authenticated channels.
We assume that the parties have established a
public-key infrastructure prior to the protocol execution.
That is, we assume that all parties hold the same vector $(pk_1,
\ldots, pk_n)$ of public keys for a digital-signature scheme,
and each honest party $P_i$ holds the honestly generated
secret key $sk_i$ associated with~$pk_i$.
A \emph{valid signature $\sigma$ on $m$ from~$P_i$} is one for
which $\vrfy_{pk_i}(m, \sigma)=1$. For readability, we use $\signed{m}{i}$ to denote a tuple
$(i,m,\sigma)$ such that $\sigma$ is a valid signature on message $m$
signed using~$P_i$'s secret key.

For simplicity, we treat signatures as ideal (i.e., perfectly unforgeable);
we also implicitly assume that
parties use some form of domain separation when signing (e.g.,  by using unique
session IDs) to ensure that signatures are valid only
in the context in which they are generated.

Where applicable, we use $\kappa$ to denote a statistical security parameter.

\medskip\noindent
\textbf{Adversarial model.} 
We consider the security of our protocols in the presence of an adversary who
can \emph{adaptively} corrupt some number of parties.
The adversary may coordinate the behavior of
corrupted parties and cause them to deviate arbitrarily from the protocol.
Note, however, that our claims about adaptive security are
only with respect to the property-based definitions found in Section~\ref{section:Definitions},
not with respect to a simulation-based definition (cf.~\cite{EC:HirZik10,PODC:GKKZ11}).

\medskip\noindent
\textbf{Network model.} We consider two possible settings for the network. In the
\emph{synchronous} case, all messages are delivered within some known
time~$\Delta$ after they are sent, but the adversary can reorder and delay
messages subject to this bound. (As a consequence, the
adversary can potentially be \emph{rushing}, i.e., it
can wait to receive all incoming messages in a round before sending its own messages.)
In this setting, we also assume 
all parties begin the protocol at
the same time, and  parties' clocks progress at the same rate.
When we
say the network is \emph{asynchronous}, we mean that the adversary can delay
messages for an arbitrarily long period
of time, though messages must
eventually be delivered.
We do not make any assumptions on parties' local clocks in the asynchronous case.

We view the network as being either synchronous or asynchronous for the lifetime of the protocol (although we stress that the honest parties do not know which is the case).

\section{Definitions}\label{section:Definitions}

Although we are ultimately interested in state machine replication, our main protocol relies on various subprotocols for different tasks. We therefore provide relevant definitions here.
Throughout, when we say a protocol achieves some property, we include the case where it achieves that property with overwhelming probability (in the implicit parameter~$\kappa$).

\subsection{Useful Subprotocols}
In some cases we consider protocols where parties may not terminate (even upon generating output); for this reason, we mention termination explicitly in some definitions.
Honest parties are those who are not corrupted by the end of the execution.


\medskip\noindent{\bf Reliable broadcast.}
A \emph{reliable broadcast} protocol allows parties to agree
on a value chosen by a designated sender.  In contrast to the stronger notion of \emph{broadcast}, here honest parties might not terminate (but, if so, then none of them terminate).

\begin{definition}[Reliable broadcast]
	Let $\Pi$ be a protocol executed by parties $P_1, \ldots, P_n$,
	where a designated sender $P^*\in\{P_1, \ldots, P_n\}$ begins holding input~$v^*$ and parties terminate upon generating output.
	\begin{itemize}

		\item {\bf Validity:} $\Pi$ is {\sf $t$-valid} if the
		following holds whenever at most $t$ parties are
		corrupted: if $P^*$ is honest, then every honest party outputs~$v^*$.
		
		\item {\bf Consistency:}  $\Pi$ is {\sf $t$-consistent} if the following holds whenever
		at most $t$ parties are corrupted: either no honest party outputs anything, or all honest parties output the same value~$v\in\bool$.

	\end{itemize}
	If $\Pi$ is $t$-valid and $t$-consistent, then we say it is {\sf $t$-secure}.
\end{definition}

\medskip\noindent{\bf Byzantine agreement.}
A \emph{Byzantine agreement} protocol allows parties who each hold some
initial value to agree on an output value. 

\begin{definition}[Byzantine agreement] \label{def:BA}
	Let $\Pi$ be a protocol executed by parties $P_1, \ldots, P_n$,
	where each party $P_i$ begins holding input $v_i \in \bool$.
	\begin{itemize}		
		\item {\bf Validity:} $\Pi$ is {\sf $t$-valid} if the
		following holds whenever at most $t$ of the parties are
		corrupted: if every honest party's input is equal to
		the same value~$v$, then every honest party
		outputs~$v$.
		
		\item {\bf Consistency:} $\Pi$ is {\sf $t$-consistent} if
		the following holds whenever at most $t$ of the parties
		are corrupted: every honest party outputs the same value~$v\in \bool$.
		
		\item {\bf Termination:} $\Pi$ is {\sf $t$-terminating} if whenever
		at most $t$  parties are corrupted, every honest
		party terminates with some output in~$\bool$.
		
	\end{itemize}
	If $\Pi$ is $t$-valid, $t$-consistent, and $t$-terminating, then we say it is {\sf $t$-secure}.
\end{definition}

\medskip\noindent{\bf Asynchronous common subset (ACS).}
Informally, a protocol for the \emph{asynchronous common subset}  problem~\cite{PODC:BenKelRab94}
allows $n$ parties, each with some input, to agree on a
subset of those inputs. (The term ``asynchronous'' in the name
is historical, and one can also consider protocols for this task in the synchronous setting.)

\begin{definition}[ACS]
	Let $\Pi$ be a protocol executed by parties $P_1, \ldots, P_n$,
	where each $P_i$ begins holding input $v_i \in \bool^*$, and parties output sets of size at most~$n$.
	\begin{itemize}
		\item{\bf Validity:} $\Pi$ is {\sf $t$-valid} if the following holds whenever at most $t$
		parties are corrupted: if every honest party's input is equal to the same value~$v$, then every honest party outputs~$\{v\}$.
		
		\item {\bf Liveness:} $\Pi$ is {\sf $t$-live} if
		whenever at most $t$ of the parties are corrupted,
		every honest party produces output.
		
		\item {\bf Consistency:} $\Pi$ is {\sf $t$-consistent} if whenever at most $t$ parties are corrupted, all honest parties output the same set~$S$. 
		
		\item{\bf Set quality:}  $\Pi$ has {\sf $t$-set quality} if
		the following holds whenever at most $t$
		parties are corrupted: if an honest party outputs a set~$S$, then $S$ contains the inputs of at least $t+1$ honest parties.
	\end{itemize}
\end{definition}


\subsection{State Machine Replication}


Protocols for \emph{state machine replication} (SMR) allow parties to maintain agreement on an ever-growing, ordered sequence of \emph{blocks}, where a block is a set of values called \emph{transactions}. An SMR protocol does not terminate but instead continues indefinitely.
We model the sequence of blocks output by a party $P_i$ via a write-once array $\blocks_i=\block{i}{1}, \block{i}{2}, \ldots$ maintained by~$P_i$,
each entry (or \emph{slot}) of which is initially equal to~$\perp$.
We say that $P_i$ \emph{outputs a block in slot~$j$} when $P_i$  writes a block to~$\block{i}{j}$;
if $\block{i}{j} \neq \perp$ then we call $\block{i}{j}$ the \emph{block output by $P_i$ in slot~$j$}.


It is useful to
define a notion of \emph{epochs} for each party. (We stress that these are not global epochs; instead, each party maintains a local view of its current epoch.) Formally, we assume that each party $P_i$ maintains a write-once array $\epochs_i=\epochs_i[1], \epochs_i[2], \ldots$, each entry of which is initialized to~0. We say $P_i$ \emph{enters epoch~$j$} when it sets $\epochs_i[j]:=1$, and
require:
\begin{itemize}
	\item For $j>1$, $P_i$ enters epoch~$j-1$ before entering epoch~$j$.
	\item $P_i$ enters epoch~$j$ before outputting a block in slot~$j$.
\end{itemize}

An SMR protocol is run in a setting where parties asynchronously receive inputs (i.e., transactions) as the protocol is being executed; each party $P_i$ stores transactions it receives in a local buffer~$\buf_i$.
We imagine these transactions
as being provided to parties by some mechanism external to the protocol (which could involve a gossip protocol run among the parties themselves), and make no assumptions about the arrival times of these transactions at any of the parties.

%

\begin{definition}[State machine replication]
	Let $\Pi$ be a protocol executed by parties $P_1, \ldots, P_n$ who are provided with transactions as input and locally maintain arrays $\blocks{}{}$ and $\epochs$ as described above.
	\begin{itemize}
		\item {\bf Consistency:} $\Pi$ is {\sf $t$-consistent} if the following holds whenever at most $t$ parties are corrupted: if an honest party outputs a block $B$ in slot~$j$ then all parties that remain honest output $B$ in slot~$j$.
		
		\item {\bf Strong liveness:} $\Pi$ is {\sf $t$-live}
		if the following holds whenever at most~$t$ parties are corrupted: for any transaction $\tx$ for which every honest party received~$\tx$ before entering epoch~$j$,
		every party that remains honest outputs a block that contains~$\tx$ in some slot~$j' \leq j$ .
		
		\item  {\bf Completeness:} $\Pi$ is {\sf $t$-complete}
		if the following holds whenever at most $t$ parties are corrupted: for all~$j\geq 0$, every party that remains honest outputs some block in slot~$j$.
	\end{itemize}
	If $\Pi$ is $t$-consistent, $t$-live, and $t$-complete, then we say it is {\sf $t$-secure}.
\end{definition}

Our liveness definition is stronger than usual, in that we require a transaction $\tx$ that appears in all honest parties' buffers by epoch~$j$ to be included in a block output by each honest party in some slot~\mbox{$j'\leq j$}. (Typically, liveness only requires that each honest party eventually output a block containing~$\tx$.)
This stronger notion of liveness is useful for showing that SMR implies Byzantine agreement (cf.\ Appendix~\ref{sec:SMR-BA}) and
is achieved by our protocol.

In our definition, a transaction $\tx$ is only guaranteed to be contained in a block output by an honest party if \emph{all} honest parties receive~$\tx$ as input. A stronger definition would be to require this to hold even if only a \emph{single} honest party receives $\tx$ as input. It is easy to achieve the latter from the former, however, by simply having honest parties gossip all transactions they receive to the rest of the network.

Our definition does not require that honest parties output a block in slot \mbox{$j-1$} before outputting a block in slot~$j$. 
If this behavior is undesirable, one could instruct each party to withhold outputting a block in slot~$j$ until it outputs blocks in all slots prior to~$j$. Any protocol secure with respect to our definition would remain secure if modified in this way.

\section{An ACS Protocol with Higher Validity Threshold}\label{sec:ACS-all}
Throughout this section, we assume an asynchronous network.
We construct an ACS protocol that is secure when the number of corrupted parties is below one threshold, and provides validity even for some higher corruption threshold. That is, fix $t_a \leq t_s$ with $t_a + 2\cdot t_s < n$. We show an ACS protocol that is $t_a$-secure, and achieves validity even for $t_s$ corruptions. This protocol will be a key ingredient in our SMR protocol.

Our construction follows the high-level approach taken by Miller et al.~\cite{CCS:MXCSS16}, who devise an ACS protocol based on subprotocols for reliable broadcast and Byzantine agreement. In our case we need a reliable broadcast protocol that achieves validity for $t_s\geq n/3$ faults, and in Section~\ref{sec:Bracha} we show such a protocol.
We then describe and analyze our ACS protocol in Section~\ref{sec:ACS}.

\subsection{Reliable Broadcast with Higher Validity Threshold}
\label{sec:Bracha}
In Figure~\ref{prot:mbb}, we present a variant of Bracha's (asynchronous) reliable broadcast protocol~\cite{PODC:Bracha84} that allows for a more general tradeoff between consistency and validity. Specifically, the protocol is parameterized by a threshold~$t_s$; for any
$t_a\leq t_s$ with $t_a + 2\cdot t_s < n$, the protocol
achieves $t_a$-consistency and $t_s$-validity.

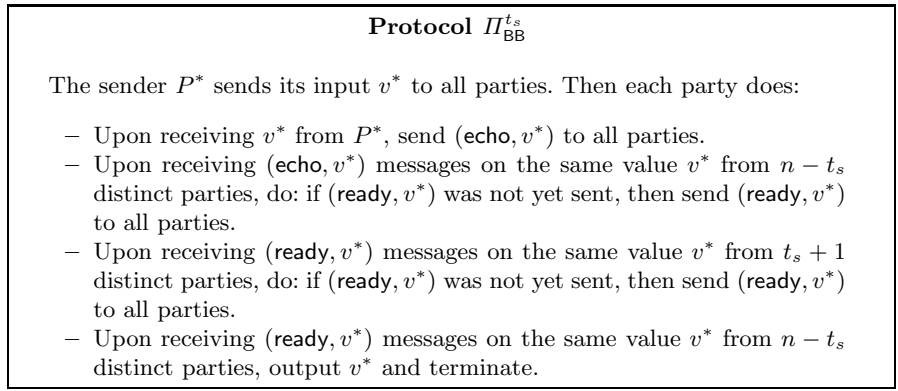
\protocol{Protocol $\BB$}{Bracha's reliable broadcast protocol, parameterized by~$t_s$. }{prot:mbb}{
	%
	The sender $P^*$ sends its input $v^*$ to all parties. Then each party does:
	\begin{itemize}
		\item Upon receiving $v^*$ from $P^*$, send $(\echo, v^*)$ to all parties.
		\item Upon receiving $(\echo, v^*)$ messages on the same value $v^*$ from $n-t_s$ distinct parties, do: if $(\ready, v^*)$ was not yet sent, then send $(\ready, v^*)$ to all parties.
		\item Upon receiving $(\ready,v^*)$ messages on the same value $v^*$ from $t_s+1$ distinct parties, do: if $(\ready, v^*)$ was not yet sent, then send $(\ready, v^*)$ to all parties.
		\item Upon receiving $(\ready, v^*)$ messages on the same value $v^*$ from $n-t_s$ distinct parties, output $v^*$ and terminate.
	\end{itemize}
}

\begin{lemma} If $t_s < n/2$ then $\BB$ is $t_s$-valid.
\end{lemma}
\begin{proof}
	Assume there are at most $t_s$ corrupted parties, and the sender is honest. All honest parties receive the same value~$v^*$ from the sender, and
	consequently send $(\echo, v^*)$ to all other parties. Since there
	are at least $n-t_s$ honest parties, all honest parties receive
	$(\echo, v^*)$ from
	at least $n-t_s$ different parties,
	and as a result send $(\ready, v^*)$ to all other parties.
	By the same argument, all honest parties receive $(\ready, v^*)$ from at least $n-t_s$ parties, and so can output~$v^*$
	(and terminate).
	
	Fix any $v \neq v^*$.
	To complete the proof, we argue that no
	honest party will output~$v$.
	Note first that no honest party will send $(\echo, v)$.
	Thus, any honest party will receive $(\echo,v)$ from at most $t_s$ other parties. Since $t_s < n-t_s$, no honest party will ever send $(\ready, v)$. By the same argument, this shows that honest parties will receive $(\ready,v)$ from at most $t_s$ other parties, and hence will not output~$v$.   \qed
\end{proof}

\begin{lemma}
	Fix $t_a\leq t_s$ with $t_a + 2\cdot t_s<n$. Then $\BB$ is $t_a$-consistent.
\end{lemma}
\begin{proof}
	Suppose at most $t_a$ parties are corrupted, and that an honest party $P_i$ outputs~$v$. Then $P_i$ must have received $(\ready, v)$ messages from at least
	$n-t_s$ distinct parties,
	at least $n-t_s-t_a \geq t_s+1$ of whom are honest. 	
	Thus, all honest parties
	receive $(\ready, v)$ messages from at least $t_s+1$ distinct parties, and so all honest parties send $(\ready, v)$ messages to everyone. It follows that all honest parties receive $(\ready, v)$ messages from at least $n-t_a\geq n-t_s$ parties, and so can output~$v$ as well.
	
	To complete the proof, we argue that
	honest parties cannot output $v' \neq v$. We argued above that all honest parties send $(\ready, v)$ to everyone. Let $P$ be the first honest party to do so. Since $t_a<t_s+1$, that party must have sent $(\ready, v)$ in response to receiving $(\echo, v)$ messages from at least $n-t_s$ distinct parties. If some honest $P_j$ outputs $v'$ then, arguing similarly, some honest party $P'$ must have received $(\echo, v')$ messages from at least $n-t_s$ distinct parties. But this is a contradiction, since honest parties send only a single $\echo$ message but $2\cdot (n-t_s) -t_a > n$.  \qed
\end{proof}

\subsection{ACS with Higher Validity Threshold}
\label{sec:ACS}
\def\bb{{\sf Bcast}}

In Figure~\ref{prot:opt}
we describe an ACS protocol $\ACS$ that is parameterized by thresholds~$t_a,t_s$, where $t_a\leq t_s$ and
$t_a + 2\cdot t_s < n$.
Our protocol relies on two subprotocols: a reliable broadcast protocol $\bb$ that is $t_s$-valid and $t_a$-consistent (such as the protocol $\BB$ from the previous section), and a Byzantine agreement protocol $\aba$ that is $t_a$-secure (since $t_a < n/3$, any asynchronous BA protocol secure for that threshold can be used).
Our ACS protocol runs several executions of these protocols as sub-routines, so to distinguish between them we denote the $i$th execution by $\bb_i$, resp.,~$\aba_i$, and say that these executions \emph{correspond to party~$P_i$}.


\protocol{Protocol $\ACS$}{An ACS protocol, parameterized by~$t_a$ and~$t_s$.}{prot:opt}{
	At any point during a party's execution of the protocol, let \mbox{$S^*\bydef \{i : \aba_i$ output 1$\}$} and let $s=|S^*|$. Define the following boolean conditions:
	\begin{itemize}
		\item $C_1(v)$: at least $n-t_s$ executions $\{\bb_i\}_{i \in [n]}$ have output~$v$.
		\item $C_1$: $\exists v$ for which $C_1(v)$ is true.
		\item $C_2(v)$: $s \geq n-t_a$, all executions $\{\aba_i\}_{i \in [n]}$ have terminated, and a majority of the executions
		$\{\bb_i\}_{i \in S^*}$ have output~$v$.
		\item $C_2$: $\exists v$ for which $C_2(v)$ is true.
		\item $C_3$: $s \geq n-t_a$, all executions $\{\aba_i\}_{i \in [n]}$ have terminated,
		and all executions $\{\bb_i\}_{i \in S^*}$ have terminated.
	\end{itemize}
	
	Each party does:
	\begin{itemize}
		\item For all $i$: run $\bb_i$ with $P_i$ as the sender, where $P_i$ uses input~$v_i$.
		\item When $\bb_i$ terminates with output $v'_i$ do: if execution of $\aba_i$ has not yet begun, 
		run $\aba_i$ using input~$1$.
		
		\item When $s \geq n-t_a$,
		run any executions $\{\aba_i\}_{i \in [n]}$ that have not yet begun, using input~0.

		
		\item \textbf{(Exit 1:)} If at any point $C_1(v)$ for some~$v$, output~$\{v\}$. 
		
		\item \textbf{(Exit 2:)} If at any point $\neg C_1 \wedge C_2(v)$ for some~$v$, output~$\{v\}$.
		
		\item\textbf{(Exit 3:)} If at any point $\neg C_1 \wedge \neg C_2 \wedge C_3$,
		output $S:=\{v'_i\}_{i \in S^*}$.
	\end{itemize}
	
	After outputting:
	\begin{itemize}
		\item Continue to participate in any ongoing $\bb$ executions.
		\item Once $C_1=\true$, stop participating in any ongoing $\aba$ executions.
	\end{itemize}
}





\begin{lemma}\label{lem:acs_validity} If $t_a + 2 \cdot t_s < n$,
	then $\ACS$ is $t_s$-valid. 
\end{lemma} 
\begin{proof}
	Note that $t_s<n/2$. Say at most $t_s$ parties are dishonest, and all honest parties have the same input~$v$.
	By $t_s$-validity of $\bb$,
	at least $n-t_s$ executions of~$\{\bb_i\}$ (namely, those for which $P_i$ is honest) will result in $v$ as output, and so all honest parties can take Exit~1 and output~$\{v\}$.
	It is not possible for an honest party to take Exit~1 and output something other than~$\{v\}$, since $t_s < n-t_s$.
	Thus, it only remains to show that if an honest party takes some other exit then it must also output~$\{v\}$.
	Consider the two possibilities:
	
	\medskip\noindent{\bf Exit 2:} Suppose some honest party $P$ takes Exit~2 and outputs~$\{v'\}$.
	Then, for that party, $C_2(v')$ is true, and so $P$  must have seen at least
	$\lfloor \frac{s}{2}\rfloor +1$ of the $\{\bb_i\}_{i \in S^*}$ terminate with output~$v'$. Moreover, $P$ must have $s \geq n-t_a$.
	Together, these imply that $P$ has seen at least
	$$\left\lfloor \frac{n-t_a}{2}\right\rfloor +1\geq \left\lfloor \frac{2t_s}{2}\right\rfloor+1>t_s$$
	executions of $\{\bb_i\}$ terminate with output~$v'$.
	At least one of those executions must correspond to an honest party. But then $t_s$-validity of $\bb$ implies that $v'=v$.
	
	\medskip\noindent{\bf Exit 3:} Assume an honest party $P$ takes Exit~3.
	Then $P$ must have $s \geq n-t_a$,
	must have seen all executions $\{\aba_i\}_{i \in [n]}$ terminate, and
	must also have
	seen all executions $\{\bb_i\}_{i \in S^*}$ terminate.
	Because
	\[|S^*|=s\geq n-t_a > 2t_s,\]
	a majority of the executions $\{\bb_i\}_{i \in S^*}$ that $P$ has seen terminate must correspond to honest parties.
	By $t_s$-validity of $\bb$, 
	all those executions must have resulted in output~$v$. But then $C_2(v)$ must be true for $P$, and it would not have taken Exit~3.  \qed
\end{proof}

\begin{lemma}\label{acs_equal}
	Fix $t_a \leq t_s$ with $t_a + 2\cdot t_s < n$, and say at most $t_a$ parties are corrupted.
	If honest parties $P_1, P_2$ output sets $S_1, S_2$, then $S_1=S_2$.
\end{lemma} 
\begin{proof}
	We consider different cases based on the possible exits taken by $P_1$ and~$P_2$, and show that
	in all cases their outputs agree.
	
	\medskip\noindent{\bf Case 1:} \emph{Either $P_1$ or $P_2$ takes Exit~1.} Without loss of generality, assume~$P_1$ takes Exit~1 and outputs~$\{v_1\}$. We consider different sub-cases:
	\begin{itemize}
		\item {\em $P_2$ takes Exit~1:} Say $P_2$ outputs~$\{v_2\}$.  Then $P_1$ and $P_2$ must have each seen at
		least $n-t_s$ executions of $\{\bb_i\}$ output $v_1$ and $v_2$, respectively.
		Since $t_s < n/2$, at least one of those executions
		must be the same. But then $t_a$-consistency of $\bb$ implies that $v_1=v_2$.
		
		\item {\em $P_2$ takes Exit~2:} Say $P_2$ outputs $\{v_2\}$. For $C_2(v_2)$ to be satisfied,
		$P_2$ must have $s \geq n-t_a$, and
		must have seen at least
		\[\left\lfloor \frac{s}{2}\right\rfloor + 1 \geq \left\lfloor \frac{n-t_a}{2}\right\rfloor +1\]
		executions of $\{\bb_i\}$ output $v_2$. As above, $P_1$ must
		have seen at least $n-t_s$ executions of $\{\bb_i\}$ output~$v_1$. But since
		\[(n-t_s) + \left\lfloor \frac{n-t_a}{2}\right\rfloor +1\geq n-t_s + \left\lfloor \frac{2t_s}{2}\right\rfloor+1>n,\]
		at least one of those executions
		must be the same. But then $t_a$-consistency of $\bb$ implies that $v_1=v_2$.
		
		\item {\em $P_2$ takes Exit~3:} We claim this cannot occur. Indeed, if $P_2$ takes Exit~3 then $P_2$
		must have $s \geq n-t_a$, and must have
		seen all executions $\{\aba_i\}_{i \in [n]}$ terminate and all executions
		$\{\bb_i\}_{i \in S^*}$ terminate.
		Because $P_1$ took Exit~1, $P_1$ must have seen at least $n-t_s$ executions $\{\bb_i\}_{i \in [n]}$
		output~$v_1$, and therefore (by $t_a$-consistency of~$\bb$) there are at most $t_s$ executions
		$\{\bb_i\}_{i \in [n]}$ that $P_2$ has seen terminate with a value other than~$v_1$. The number of executions of
		$\{\bb_i\}_{i \in S^*}$ that $P_2$ has seen terminate with output~$v_1$ is therefore at least
		$(n-t_a) - t_s > t_s$, which is strictly greater than the number of executions $\{\bb_i\}_{i \in S^*}$ that $P_2$ has
		seen terminate with a value other than~$v_1$. But then $C_2(v_1)$ is true for $P_2$, and it would not take
		Exit~3.
	\end{itemize}

	\medskip\noindent{\bf Case 2:} \emph{Neither $P_1$ nor $P_2$ takes Exit~1.}
	We consider two sub-cases:
	\begin{itemize}
		\item {\em $P_1$ and $P_2$ both take Exit~2.} Say $P_1$ outputs~$\{v_1\}$ and
		$P_2$ outputs~$\{v_2\}$. 	
		Both $P_1$ and $P_2$ must have seen all executions $\{\aba_i\}_{i \in [n]}$ terminate; by
		$t_a$-consistency of $\aba$ they must therefore hold the same~$S^*$. Since $C_2(v_1)$ holds for $P_1$, it must have seen a majority of the executions $\{\bb_i\}_{i \in S^*}$ output~$v_1$; similarly, $P_2$ must have seen a majority of the executions $\{\bb_i\}_{i \in S^*}$ output~$v_2$. Then $t_a$-consistency of $\bb$ implies $v_1=v_2$.

		\item {\em Either $P_1$ or $P_2$ takes Exit~3.} Say $P_1$ takes Exit~3.
		(The case where $P_2$ takes Exit~3 is symmetric.)
		As above, $P_1$ and $P_2$ agree on~$S^*$ (this holds regardless of whether $P_2$ takes Exit~2 or Exit~3). Since $C_3$ holds for $P_1$ but $C_2$ does not, $P_1$ must have seen all executions $\{\bb_i\}_{i \in S^*}$ terminate but without any value being output by a majority of those executions. But then $t_a$-consistency of $\bb$ implies that $P_2$ also does not see any value being output by a majority of those executions, and so will not take Exit~2.
		Since $P_2$ instead must take Exit~3, it must have seen all executions $\{\bb_i\}_{i \in S^*}$ terminate; $t_a$-consistency of $\bb$ then implies that $P_2$ outputs the
		same set as~$P_1$.
	\end{itemize}
	This completes the proof. \qed
\end{proof}


\begin{lemma}\label{acs_liveness}
	Fix $t_a \leq t_s$ and $t_a + 2\cdot t_s < n$. Then $\ACS$ is $t_a$-live.
\end{lemma}
\begin{proof}
	If some honest party $P$ takes 
	Exit~1 during an execution of~$\ACS$, then $P$ must have seen at least $n-t_s$ executions $\{\bb_i\}_{i \in [n]}$ with the same output~$v$. By $t_a$-consistency of $\bb$, all other honest parties will eventually see at least those $n-t_s$ executions output~$v$, and will generate output (if they have not already generated output via another exit). 
	
	It remains to consider the case where no honest parties take Exit~1.
	Let $H$ be the indices of parties who remain honest, with $|H| \geq n-t_a$.
	By $t_s$-validity of $\bb$, 
	all honest parties see the executions $\{\bb_i\}_{i \in H}$ terminate, and so 
	all honest parties initiate the executions $\{\aba_i\}_{i \in H}$.
	Since no honest party takes Exit~1, all honest parties continue to participate in all those executions. Consider some execution~$\aba_i$ being run by all honest parties.
	As long as no honest party has $s \geq n-t_a$, each honest party must be running $\aba_i$ using input~1. By $t_a$-validity of $\aba$, this means that all honest parties will eventually output~1 from that execution.
	We conclude from this that some honest party will eventually have $s \geq n-t_a$; furthermore, $t_a$-consistency of $\aba$ then implies that all honest parties will eventually have $s \geq n-t_a$. This means that all honest parties execute all $\{\aba_i\}_{i \in [n]}$, and by $t_a$-security of $\aba$ all those executions eventually terminate.
	Define 
	$\hat S^* \bydef \{i : 
	\mbox{some honest player outputs 1 in $\aba_i$}\}$.
	We claim that all executions $\{\bb_i\}_{i \in \hat S^*}$ eventually terminate. To see this, fix $i \in \hat S^*$. Then by $t_a$-validity of $\aba$, some honest party $P$ must have used input~1 to~$\aba_i$. But that implies that $\bb_i$ must have terminated for~$P$. So $t_a$-consistency of $\bb$ implies that $\bb_i$ will terminate for all honest parties.
	It follows that any honest party can take Exit~2 or~3. \qed
\end{proof}

\begin{lemma}\label{lem:acs_setquality}
	Fix $t_a \leq t_s$ with $t_a + 2\cdot t_s < n$. Then $\ACS$ has $t_a$-set quality.
\end{lemma}
\begin{proof}
	Consider some honest party $P$. Say $P$ takes Exit~1 and outputs~$S=\{v\}$. Then $P$ has seen at least $n-t_s$ executions $\{\bb_i\}$ terminate with output~$v$.
	Of these, at least $n-t_s-t_a > t_s \geq t_a$ must correspond to honest parties.
	By $t_s$-validity of $\bb$, those honest parties all had input~$v$.
	This means that $S$ contains the inputs of at least $t_a+1$ honest parties.
	
	Alternatively, say $P$ takes Exit~2 or~3 and outputs a set~$S$.
	Then $P$ must have $|S^*|\geq n-t_a$. At least
	\[n-2\cdot t_a > \max\{(n-t_a)/2, \; t_a  \}\]
	of the indices in $S^*$ correspond to honest parties, and $t_s$-validity of $\bb$ implies that for each of those parties
	the corresponding output value $v'_i$ that $P$ holds is equal to that party's input. Thus, regardless of whether $P$ takes Exit~2 (and $S$ contains the majority value output by $\{\bb_i\}_{i \in S^*}$) or Exit~3 (and $S$ contains every value output by $\{\bb_i\}_{i \in S^*}$), the set $S$ output by $P$ contains the inputs of at least $t_a+1$ honest parties. \qed
\end{proof}

\begin{theorem}\label{thm:acs_consistency}
	Fix $t_a, t_s$ with $t_a \leq t_s$ and $t_a + 2\cdot t_s < n$. Then $\ACS$ is $t_a$-secure and $t_s$-valid.
\end{theorem}
\begin{proof}
	Lemma~\ref{lem:acs_validity} proves $t_s$-validity. Lemmas~\ref{acs_equal} and \ref{acs_liveness} together prove $t_a$-consistency, and  Lemma~\ref{lem:acs_setquality} shows $t_a$-set quality. \qed
\end{proof}

\begin{lemma}\label{lem:acs:cc}
	Fix $t_a \leq t_s$ with $t_a + 2\cdot t_s < n$. Then $\ACS$ has bounded
	communication complexity under either of the following conditions:
	\begin{enumerate}
		\item At most $t_a$ parties are corrupted.
		\item At most $t_s$ parties are corrupted and all honest parties have the same input.
	\end{enumerate}
\end{lemma}
\begin{proof}
	Because $\bb$ from the previous section has bounded communication complexity, we only need to show that all honest parties eventually stop participating in all $\aba$ executions. (This can occur either because those executions all terminate, or because honest parties all set $C_1=\true$ and stop participating in any still-running
	executions.)
	
	\medskip\noindent{\bf Case 1:} \emph{At most $t_a$ parties are corrupted.} 
	If some honest party~$P$ takes Exit~1 during an execution of $\ACS$, then $P$ must have seen at least $n-t_s$ executions $\{\bb_i\}_{i \in [n]}$ with the same output value. By $t_a$-consistency of $\bb$, all honest parties eventually see those executions output the same value, and thus set $C_1=\true$ and stop participating in any still-running $\aba$ executions.
	
	On the other hand, if no honest parties take Exit~1 during an execution of $\aba$, then all honest parties continue to participate in all $\aba$ executions. By $t_a$-termination of~$\aba$, each of those executions will terminate.

	\medskip\noindent{\bf Case 2:}  \emph{At most $t_s$ parties are corrupted and all honest parties have the same input~$v$.} 
	Because all honest parties have input~$v$, $t_s$-validity of $\bb$ implies that all honest parties receive output $v$ from at least $n-t_s$ executions of $\{\bb_i\}$.  So all honest parties will eventually set $C_1=\true$ and thus stop participating in any still-running $\aba$ executions. \qed
\end{proof}

\section{A Network-Agnostic SMR Protocol}\label{sec:SMR}

In this section, we show our main result: an SMR protocol that is $t_s$-secure in a synchronous network and $t_a$-secure in an asynchronous network. We begin in Section~\ref{sec:bla} by briefly introducing a useful primitive
called \emph{block agreement}. In Appendix~\ref{sec:bla-proof}, we construct a block-agreement protocol secure against $t<n/2$ parties in a synchronous network.
We then use our block-agreement protocol to construct an SMR protocol in Section~\ref{sec:SMR-main}.

\subsection{Block Agreement}
\label{sec:bla}

Block agreement is a form of agreement where (1)~in addition to an input, parties provide signatures (in a particular format) on those inputs, and (2)~a stronger notion of validity is required. 
Specifically, 
consider pairs consisting of a block $B$ along with a set~$\Sigma$ of signed buffers $\signed{\buf_j}{j}$.
(Recall that $\signed{m}{i}$ denotes a tuple
$(i,m,\sigma)$ such that $\sigma$ is a valid signature on~$m$
with respect to~$P_i$'s secret key.)
We say a pair $(B, \Sigma)$ is \emph{$t$-valid} if:
\begin{itemize}
	\item $\Sigma$ contains signed buffers from strictly more than $t$ distinct parties.
	\item For each $\signed{\buf_j}{j}\in\Sigma$, we have $\buf_j \subseteq B$. 
\end{itemize}
A pair is \emph{valid} if it is 0-valid (meaning it contains signed buffers from at least one party).

\begin{definition}[Block agreement]
	Let $\Pi$ be a protocol executed by parties $P_1, \ldots, P_n$,
	where each party $P_i$ begins holding a valid pair $(B_i, \Sigma_i)$ and parties terminate upon generating output.
	\begin{itemize}		
		
		\item {\bf Validity:} $\Pi$ is {\sf $t$-valid} if whenever at most $t$ of the parties are corrupted, every honest party
		outputs a valid pair.
		
		
		\item {\bf Termination:} $\Pi$ is {\sf $t$-terminating} if
		whenever at most $t$ of the parties
		are corrupted, every honest party terminates.
		
		\item {\bf Consistency:} $\Pi$ is {\sf $t$-consistent} if
		the following holds whenever at most $t$ of the parties
		are corrupted: for any $s \leq t$, if every honest party inputs an $s$-valid pair, 
		there is an $s$-valid $(B, \Sigma)$ such that every honest party outputs $(B, \Sigma)$.
	\end{itemize}
	If $\Pi$ is $t$-valid, $t$-consistent, and $t$-terminating, then we say it is {\sf $t$-secure}.
\end{definition}

We prove the following in Appendix~\ref{sec:bla-proof}.

\begin{theorem}\label{thm:exists-bla}
	There
	is a block-agreement protocol $\BLA$ that is $t$-secure for any $t<n/2$ when run in a synchronous network. Moreover, all honest parties terminate with probability $1-2^{-O(\kappa)}$ after time~$\kappa \cdot \Delta$.
\end{theorem}

\subsection{State Machine Replication}
\label{sec:SMR-main}
We now combine our various sub-protocols to realize network-agnostic SMR.
At a high level, our SMR protocol $\SMR$ (see Figure~\ref{prot:smr}) proceeds as follows. For each slot~$j$, the parties attempt to reach agreement on a block using the block-agreement protocol $\BLA$. If that protocol terminates, parties use its output $B$ as input to our ACS protocol~$\ACS$. 
If $\BLA$ fails to terminate after a sufficiently long time, parties abandon it and instead attempt to reach agreement using the ACS protocol directly.

By setting the timeout appropriately, we can ensure that in a synchronous network $\BLA$ terminates with overwhelming probability. Thus, 
if the network is synchronous and at most $t_s$ parties are corrupted, all parties agree on their input $B$ to~$\ACS$, and $t_s$-validity of $\ACS$ ensures that all parties output~$B$.
On the other hand, if the network is asynchronous and at most $t_a$ parties are corrupted, then 
$t_a$-security of $\ACS$ ensures agreement.

\protocol{Protocol $\SMR$}{A protocol for state machine replication.}{prot:smr}{We describe the protocol
	from the point of view of party $P_i$ holding a set~$\buf_i$ that grows asynchronously via some external process.
	
	\smallskip
	For $k=1, \ldots$, do the following starting at time $T_k:=(\Delta+\kappa\Delta)\cdot (k-1)$:
	\begin{enumerate}
		\item Set $\epoch{i}{k}:=1$, and initialize $B:=\emptyset,\Sigma:=\emptyset$.
		\item Send $\signed{\buf_i}{i}$ to every party.
		\item While $|\Sigma|\leq t_s$:
		\begin{itemize} 
			\item The first time $\signed{\buf_j}{j}$ is received from $P_j$, set $B:=B \cup \buf_j$ and $\Sigma:=\Sigma \cup \{\signed{\buf_j}{j}\}$. 
		\end{itemize}
		\item At time $T_k+\Delta$, run $\BLA$ on input $(B,\Sigma)$. 
		\item If $\BLA$ produces $t_s$-valid output, let $(B^*,\Sigma^*)$ denote that output. Otherwise, at time $T_k+\Delta+\kappa \cdot \Delta$ set $(B^*,\Sigma^*):=(B,\Sigma)$.
		\item Run  $\sblock\gets\ACS$ using input $B^*$.
		\item Set $\block{i}{k}:=\bigcup_{\hat B \in \sblock} \hat B$. Set $\buf_i:=\buf_i\setminus\block{i}{k}.$
	\end{enumerate}
}

We note that 
$\ACS$ does not guarantee termination. Given that any SMR protocol must run indefinitely, however, this seems reasonable, especially since $\ACS$ has bounded communication complexity when run in the context of~$\SMR$ (cf.\ Lemma~\ref{lem:acs:cc} and the proofs below).

We now prove security of $\SMR$ in a network-agnostic setting.


\begin{theorem}[Consistency]\label{smr_consistency}
	Fix $t_a, t_s$ with $t_a < n/3$ and $t_a + 2\cdot t_s < n$. Then $\SMR$ is $t_a$-consistent when run in an asynchronous network, and $t_s$-consistent when run in a synchronous network.
\end{theorem}
\begin{proof} Assume first that at most $t_s$ parties are dishonest and the network is synchronous. In any slot~$k$, 
	each honest party receives $\signed{\buf_j}{j}$ from at least the $n-t_s>t_s$ honest parties, and the input $(B, \Sigma)$ they use to $\BLA$ is $t_s$-valid.
	Consistency of $\BLA$ implies that every honest party outputs the same $t_s$-valid pair $(B^*, \Sigma^*)$ after running $\BLA$ for time $\kappa \cdot \Delta$.
	By $t_s$-validity of $\ACS$, this means every honest
	party obtains output~$\{B^*\}$ from $\ACS$ and then sets $\block{}{k}=B^*$.
	
	If at most $t_a$ parties are dishonest and the network is asynchronous, then
	$t_a$-consistency of $\ACS$ implies that
	all honest parties agree on the same value $\sblock$, and hence set $\block{}{k}$ to the same value.
	\qed \end{proof}

\begin{theorem}[Strong liveness]\label{smr_liveness}
	Fix $t_a \leq t_s$ with $t_a + 2\cdot t_s < n$. Then $\SMR$ is $t_a$-live when run in an asynchronous network, and $t_s$-live when run in a synchronous network.
\end{theorem}
\begin{proof} 
	By consistency of $\SMR$, we can refer to the values of $\block{}{i}$ without specifying any particular party. 
	Consider some transaction $\tx$ that every honest party received before entering epoch~$k$. If $\tx$ appears in $\block{}{k'}$ for some $k' < k$ then we are done. Otherwise, every honest party has $\tx$ in their buffer when they enter epoch~$k$. We show that in this latter case, $\tx$ is in $\block{}{k}$.
	
	
	Assume at most $t_s$ parties are corrupted and the network is synchronous. 
	Reasoning as in the proof of Theorem~\ref{smr_consistency}, every honest party outputs the same $t_s$-valid pair $(B^*, \Sigma^*)$ after running $\BLA$ for time~$\kappa \cdot \Delta$, and sets $\block{}{k}=B^*$.
	Since $(B^*, \Sigma^*)$ is $t_s$-valid, $\Sigma^*$ must contain a signature on a subset of $B^*$ from at least one honest party. But an honest party would have only signed a subset that includes~$\tx$, implying $\tx \in B^*$.

	Consider next the case where at most $t_a$ parties are dishonest and the network is asynchronous. 
	Every honest party $P_i$ runs $\ACS$ using an input $B^*_i$ for which they have a $t_s$-valid pair $(B^*_i, \Sigma^*_i)$. Arguing as above, each $B^*_i$ must contain~$\tx$. By $t_a$-security of $\ACS$, all honest parties output the same set $\sblock$ that contains $B^*_i$ for some honest party~$P_i$, and hence contains~$\tx$. It follows that every honest party includes $\tx$ in $\block{}{k}$.
	\qed
\end{proof}

\begin{theorem}[Completeness] Fix $t_a, t_s$ with $t_a < n/3$ and $t_a + 2\cdot t_s < n$. 
	Then $\SMR$ is $t_a$-complete when run in an asynchronous network, and $t_s$-complete when run in a synchronous network.	
\end{theorem}
\begin{proof} 
	By inspection of $\SMR$, a party outputs a block in slot~$k$ iff its execution of $\ACS$ in iteration~$k$ produces output. So if at most $t_a$ parties are corrupted, completeness follows from $t_a$-liveness of~$\ACS$. If at most $t_s$ parties are corrupted and the network is synchronous, then consistency of $\BLA$ implies that all honest parties run $\ACS$ using the same input; completeness then follows from  $t_s$-validity of~$\ACS$.
	\qed
\end{proof}

\section{Optimality of Our Thresholds}
\label{sec:lower-bound}
In this section we show that the parameters achieved by our SMR protocol are optimal. 
This extends the analogous result by Blum et al.~\cite{BKL19}, who consider the case of BA. We remark that, although SMR is generally viewed as a stronger form of consensus than BA, it is unclear whether SMR generically implies BA in a network-agnostic setting, and we were not able to show such a result for the corruption thresholds of interest (namely, when $t_a + 2t_s \geq n)$. We thus need to prove impossibility directly.

\begin{lemma}
	Fix $t_a, t_s, n$ with $t_a + 2t_s \geq n$. If an $n$-party SMR protocol is $t_s$-live 
	in a synchronous
	network, then it cannot also be $t_a$-consistent in an
	asynchronous network.
\end{lemma}
\begin{proof}
	Assume
	$t_a+2t_s=n$ and fix an SMR protocol~$\Pi$. Partition the $n$
	parties into sets $S_0, S_1, S_a$ where $|S_0|=|S_1|=t_s$ and
	$|S_a|=t_a$, and consider the following experiment:
	\begin{itemize}
		\item Choose uniform $m_0, m_1 \in \bool^\kappa$.
		\item Parties in $S_b$ begin running $\Pi$ at global time~0 with their buffers containing only~$m_b$.
		All communication between parties in $S_0$ and parties in $S_1$ is blocked (but all other messages are delivered within time~$\Delta$).
		\item Create virtual copies of each party in $S_a$, call
		them $S_a^0$ and $S_a^1$. Parties in $S_a^b$ begin running $\Pi$ (at global time $0$)
		with their buffers containing only~$m_b$, and communicate only with each other and
		parties in~$S_b$. 
	\end{itemize}
	Consider an execution of $\Pi$ in a synchronous
	network where parties in $S_1$ are corrupted and simply
	abort. Uniform $m_0, m_1 \in \bool^\kappa$ are chosen, and the remaining (honest) parties start with their buffers containing only~$m_0$.
	The views of the honest parties in this execution are
	distributed identically to the views of $S_0 \cup
	S_a^0$ in the above experiment. In particular,
	$t_s$-liveness of~$\Pi$ implies that, in the above experiment,
	all parties in $S_0$ include~$m_0$ in $\blocks[1]$. Moreover, since parties in $S_0$ have no information about~$m_1$, they include $m_1$ in $\blocks[1]$ with negligible probability.
	Analogously, all parties in $S_1$ include~$m_1$ in $\blocks[1]$ but include $m_0$ in $\blocks[1]$ with negligible probability.
	
	Next consider an execution of $\Pi$ in an asynchronous
	network where parties in $S_a$ are corrupted, and run
	$\Pi$ with their buffers containing~$m_0$ when interacting with~$S_0$ while running $\Pi$ with their buffers containing $m_1$ when interacting with~$S_1$.
	Moreover, all communication between the
	(honest) parties in $S_0$ and $S_1$ is delayed
	indefinitely. The views of the honest parties here are distributed identically to the views of
	$S_0 \cup S_1$ in the above experiment, yet the conclusion of the preceding
	paragraph shows that
	$t_a$-consistency is violated with overwhelming probability. \qed
\end{proof}

\section*{Acknowledgments}
Work supported in part under financial assistance award 70NANB19H126 from the U.S. Department of Commerce, National Institute of Standards and Technology, and NSF award~\#1837517.

\def\shortbib{0}

\appendix
\section{SMR Implies Weak BA}
\label{sec:SMR-BA}

\def\SMR{\Pi_{\sf SMR}}
We briefly discuss how SMR relates to BA. Specifically, we show that SMR implies \emph{weak}~BA. A weak BA protocol $\Pi$ satisfies validity and consistency as in Definition~\ref{def:BA}, but instead of termination it achieves a weaker liveness property. Namely, we say that
$\Pi$ is {\sf $t$-live} if whenever
at most $t$ parties are corrupted, every honest
party outputs a value in~$\bool$ (but may not terminate).

In Figure~\ref{prot:ba_smr} we show how to use an SMR protocol $\SMR$ to achieve weak~BA.

%
%
%


\protocol{Protocol $\Pi^{t_s}_{\mathsf{WBA}}$}{A protocol for weak Byzantine agreement, parameterized by~$t_s$.}{prot:ba_smr}{We describe the protocol
	from the point of view of a party $P_i$ with input~$v_i$.

	\smallskip
	\begin{itemize}
		\item Set $V_i:=\buf_i:=\emptyset$.
		\item Send $\signed{v_i}{i}$ to every party. 
		Upon receiving 
		$\signed{v_j}{j}$  from party $P_j$, set
		$\buf_i:=\buf_i\cup\{\signed{v_j}{j}\}$.
		\item Begin to run $\SMR$ at time $\Delta$.
		\item Upon outputting a block $B=\block{i}{k}$ do:
		for all $j$ such that $B$ contains $\signed{v_j}{j}$ and there is no pair $(\star, j)$ in $V$, add $(v_j, j)$ to~$V$.
		
		\item If at any point during the execution $|V_i|\geq n-t_s$, then output the majority value among all values in~$V_i$.
	\end{itemize}
}

\def\WBA{\Pi_{\sf WBA}^{t_s}}

\begin{lemma}[Validity and liveness] \label{lem:bla:validity} Let $t_a + 2t_s < n$. If $\SMR$ is $t_s$-live in a synchronous network (resp., $t_a$-live in an asynchronous network), then $\WBA$ is $t_s$-valid and $t_s$-live in a synchronous network (resp., $t_a$-valid and $t_a$-live in an asynchronous network). 
\end{lemma}
\begin{proof}
	Assume all honest parties hold input $v$. 
	Consider first the case where at most $t_s$  parties are corrupted and the network is synchronous.
	The initial message from each honest party is received by all other honest parties by time~$\Delta$. By $t_s$-liveness of $\SMR$, the block $B=\block{}{1}$ output by any honest party contains $(v, i)$ for each honest party~$P_i$.
	At that point, each honest party will have $|V| \geq n-t_s$, and since $t_s < n/2$ the majority value in~$V$ will be~$v$.
	Thus, all honest parties output~$v$.
	
	Next, consider the case where there are at most $t_a$ corrupted parties and the network is asynchronous.
	If some honest party has $|V| \geq n-t_s$, then at least $n-t_s-t_a>t_a$ of those values correspond to honest parties, and hence $v$ will be the majority value. Thus, any honest party who outputs anything will output~$v$. 
	It remains to show that all honest parties
	eventually have $|V| \geq n-t_s$. This follows from the fact that honest parties' initial messages are eventually delivered to all honest parties, along with $t_a$-liveness of~$\SMR$.
	\qed \end{proof}

\begin{lemma}[Consistency] 
	For all $t$, if $\SMR$ is $t$-consistent in a synchronous
	(resp., asynchronous) network, then $\WBA$ is $t$-consistent in a synchronous (resp., asynchronous) network. 
\end{lemma}
\begin{proof}
	The lemma is immediate.
	\qed \end{proof}

\section{A Block-Agreement Protocol}
\label{sec:bla-proof}
Throughout this section, we assume a synchronous network.

\def\status{{\sf status}}

The structure of our block-agreement protocol is inspired by the 
\emph{synod protocol} of Abraham et al.~\cite{EPRINT:ADDNR17}.
We construct our protocol in a modular fashion.
We begin by defining a subprotocol $\prp$ (see Figure~\ref{prot:propose}) in which a designated party $P^*$ serves as a \emph{proposer}.
A tuple $(k, B, \Sigma, C)$ 
is called a \emph{$k$-vote on $(B, \Sigma)$} if $(B, \Sigma)$ is valid and either:
\begin{itemize}
	\item $k=0$, or
	\item $k>0$ and $C$ is a set of valid signatures from a majority of the parties on messages of the form $(\mathsf{Commit},k', B, \Sigma)$ with $k'\geq k$ (where possibly different $k'$ can be used in different messages).
\end{itemize}
When the exact value of $k$ is unimportant, we simply refer to the tuple as a \emph{vote}.
A message of the form $\status=\signed{\mathsf{Status}, k, B,\Sigma,C}{i}$ is a \emph{correctly formed $\mathsf{Status}$ message (from party $P_i$)} if $(k, B,\Sigma,C)$ is a vote. A message 
$\signed{ \mathsf{Propose},\status_1,\ldots}{*}$
is a \emph{correctly formed $\mathsf{Propose}$ message} if
it contains 
correctly formed $\mathsf{Status}$ messages from a majority of the parties.

\pprotocol{Protocol $\prp$}{A protocol $\prp$ with designated proposer $P^*$.}{prot:propose}{bth}{We describe the protocol from the point of view of a party $P_i$ with input a vote $(k,B,\Sigma,C)$. Let $t=\lceil (n+1)/2 \rceil$.

	\begin{enumerate}
		\item At time $0$, send $\status_i:=\signed{\mathsf{Status}, k, B,\Sigma,C}{i}$ to $P^*$. 	
		\item At time $\Delta$, if $P^*$ has received at least $s \geq t$ correctly formed $\mathsf{Status}$ messages $\status_1,\ldots,
		\status_t$ 
		(from distinct parties), then $P^*$ sets 
		\[m:=(\mathsf{Propose}, \status_1,\ldots,
		\status_s),\] and sends $\signed{m}{*}$ to all parties.
		
		\item At time $2\Delta$, if a correctly formed $\mathsf{Propose}$ message
		$\signed{m}{*}$ has been received from~$P^*$, then send $\signed{m}{*}$ to all parties. Otherwise, output~$\bot$.
		
		\item At time $3\Delta$, let $\signed{m}{*}^j$ be the correctly formed $\mathsf{Propose}$ message received from~$P_j$ (if any).
		If there exists $j$ such that $\signed{m}{*}^j\neq \signed{m}{*}$, output~$\bot$. 
		Otherwise, let $\status_{\sf max} = \signed{\mathsf{Status}, k', B', \Sigma', C'}{}$
		be the status message in $\signed{m}{*}$ with maximal $k'$
		(picking the lowest index in case of ties).
		Output $(B',\Sigma')$.
	\end{enumerate}
}

We first show that any two honest parties who generate output in this protocol agree on their output.

\begin{lemma} \label{lem:prop:cons2} If honest parties $P_i$ and $P_j$ output $(B_i,\Sigma_i), (B_j,\Sigma_j) \neq \perp$, respectively, in an execution of $\prp$, then  $(B_i,\Sigma_i)=(B_j,\Sigma_j)$.
\end{lemma}
\begin{proof} 
	If $P_i$ outputs $(B_i, \Sigma_i) \neq \perp$, then $P_i$ must have received a correctly formed $\mathsf{Propose}$ message $\signed{m}{*}$ by time~$2\Delta$ that would cause it to output $(B_i, \Sigma_i)$. That message is forwarded by $P_i$ to~$P_j$, and hence $P_j$ either outputs $\perp$ (if it detects an inconsistency) or the same value $(B_i, \Sigma_i)$. \qed
\end{proof}

Assume less than half the parties are corrupted. 
We show that if there is some $(B, \Sigma)$ such that the input of each honest party $P_i$ is a vote of the form $(k_i, B, \Sigma, C_i)$, and no honest party ever receives a vote $(k', B', \Sigma', C')$ with $k' \geq \min_i\{k_i\}$ and $(B', \Sigma') \neq (B, \Sigma)$, then the only value an honest party can output is~$(B, \Sigma)$.

\begin{lemma} \label{lem:prop:cons}
	Assume fewer than $n/2$ parties are corrupted,
	and that the input of each honest party $P_i$ to $\prp$ is a $k_i$-vote on $(B, \Sigma)$. If no honest party
	ever receives a $k'$-vote on $(B',\Sigma')\neq (B,\Sigma)$ with $k' \geq \min_i\{k_i\}$, then every honest party outputs either 
	$(B,\Sigma)$ or~$\bot$.
\end{lemma}
\begin{proof} 
	Consider an honest party $P$ who does not output~$\perp$. That party must have received a correctly formed $\mathsf{Propose}$ message $\signed{m}{*}$ from~$P^*$, 
	which in turn must contain a correctly formed $\mathsf{Status}$ message from at least one honest party~$P_i$. 
	That $\mathsf{Status}$ message contains a vote  $(k_i, B, \Sigma, C_i)$ and, 
	under the assumptions of the lemma, any other vote 
	$(k', B', \Sigma', C')$ contained in
	$\signed{m}{*}$ 
	with $k' \geq k_i$ has $(B', \Sigma')=(B, \Sigma)$. 
	It follows that $P$ outputs~$(B, \Sigma)$.  \qed
\end{proof}

Finally, we show that when $P^*$ is honest then all honest parties do indeed generate output.

\begin{lemma}  \label{lem:prop:honprop}
	Assume fewer than $n/2$ parties are corrupted. If every honest party's input to $\prp$ is a vote and $P^*$ is honest, then every honest party outputs the same valid $(B,\Sigma) \neq \perp$. 
\end{lemma}
\begin{proof}  
	Since every honest party's input is a vote, $P^*$ will receive at least $\lceil (n+1)/2 \rceil$ correctly formed $\mathsf{Status}$ messages, and so sends a correctly formed $\mathsf{Propose}$ message to all honest parties.
	Since $P^*$ is honest, this is the only correctly formed $\mathsf{Propose}$ message the honest parties will receive, and so all honest parties 	
	will output the same valid $(B, \Sigma) \neq \perp$.
	\qed \end{proof}


We now present a protocol $\GC$ that uses $\prp$ to achieve a form of graded consensus on a valid pair $(B,\Sigma)$. (See Figure~\ref{prot:gc}.)
As in the protocol of Abraham et al.~\cite{EPRINT:ADDNR17}, we rely on an atomic leader-election mechanism $\leader{}$ with the following properties: On input~$k$ from a majority of parties, $\leader{}$ chooses a uniform leader $\ell \in \{1, \ldots, n\}$ and sends $(k,\ell)$ to all parties. This ensures that if less than half of all parties are corrupted, then at least one honest party must call $\leader{}$ with input~$k$ before the adversary can learn the identity of~$\ell$. 
A leader-election mechanism tolerating any $t < n/2$ faults can be realized (in the synchronous model with a PKI) based on general assumptions~\cite{DBLP:journals/jcss/KatzK09}; it can also be realized more efficiently using a threshold unique signature scheme.

Below, we refer to a message $\signed{\mathsf{Commit},k, B,\Sigma}{i}$ as a \emph{correctly formed $\mathsf{Commit}$ message {\rm (}from $P_i$ on $(B, \Sigma)$}) if 
$(B,\Sigma)$
is valid.
We refer to a message $(\mathsf{Notify},k, B,\Sigma, C)$ as a \emph{correctly formed $\mathsf{Notify}$ message on $(B, \Sigma)$} if $(B,\Sigma)$ is valid
and $C$ is a set of valid signatures on $(\mathsf{Commit},k,B,\Sigma)$ from more than $n/2$  parties; 
in that case, $C$ is called a \emph{$k$-certificate for}  
$(B, \Sigma)$. For an output $((B, \Sigma, C), g)$, we refer to $g$ as the \emph{grade} and $(B, \Sigma, C)$ as the 
\emph{output}. When a party's output is $(B, \Sigma, C)$, we may also say that its output is a $k$-certificate for $(B, \Sigma)$.

\protocol{Protocol $\GC$}{A graded block-consensus protocol $\GC$, parameterized by~$k$.}{prot:gc}{We describe the protocol from the point of view of a party $P_i$ with input a vote 
	$(k',B,\Sigma,C')$. Let $t=\lceil (n+1)/2 \rceil$.
	\begin{enumerate}
		\item At time $0$, run parallel executions of $\Pi_{\mathsf{Propose}}^{P_1},\ldots,\Pi_{\mathsf{Propose}}^{P_n}$, each using input $(k',B,\Sigma,C')$. Let $(B_j,\Sigma_j)$ be the output from the $j$th protocol.
		
		\item At time $3\Delta$, call $\leader{}(k)$ to obtain the response~$\ell$. If $(B_\ell,\Sigma_\ell)\neq\bot$, send $\langle\mathsf{Commit},k, B_{\ell},\Sigma_{\ell}\rangle_i$ to every party. 
		
		\item At time $4\Delta$, if at least $t$ correctly formed $\mathsf{Commit}$ messages $\langle\mathsf{Commit},k, B_{\ell},\Sigma_{\ell}\rangle_j$ from distinct parties have been received, then form a $k$-certificate $C$ for $(B_{\ell},\Sigma_{\ell})$, send $m:=(\mathsf{Notify},k, B_\ell,\Sigma_\ell, C)$ to every party, output $((B_\ell,\Sigma_\ell, C),2)$, and terminate.
		
		\item At time $5\Delta$, if a correctly formed $\mathsf{Notify}$ message $(\mathsf{Notify},k, B,\Sigma, C)$ has been received, output $((B,\Sigma, C),1)$ and terminate. (If there is more than one such message, choose arbitrarily.) Otherwise, output $(\bot,0)$ and terminate.
	\end{enumerate}
}



\begin{lemma} \label{lem:gc:val2}
	Assume fewer than $n/2$ parties are corrupted, and that the input of each honest party $P_i$ to $\GC$ is a $k_i$-vote on $(B, \Sigma)$. If no honest party ever receives a $k'$-vote on $(B',\Sigma')\neq (B,\Sigma)$ with $k' \geq \min_i\{k_i\}$ in step~1 of~$\GC$, then (1)~no honest party sends a $\textsf{Commit}$ message on $(B', \Sigma') \neq (B, \Sigma)$ and (2)~any honest party who outputs a nonzero grade outputs 
	a $k$-certificate for~$(B,\Sigma)$. 
\end{lemma}
\begin{proof} 
	By Lemma~\ref{lem:prop:cons}, every honest party outputs either $(B,\Sigma)$ or $\perp$ in every execution of~$\Pi_{\mathsf{Propose}}$ in step~1. It follows that 
	no honest party $P_i$ sends a $\mathsf{Commit}$ message 
	on $(B',\Sigma')\neq (B,\Sigma)$, proving the first part of the lemma.
	Since less than half the parties are corrupted, this means an honest party will receive fewer than~$\lceil (n+1)/2 \rceil$ correctly formed $\mathsf{Commit}$ 
	messages on anything other than $(B,\Sigma)$; it follows that if an honest party outputs grade $g=2$ then that party outputs $(B, \Sigma, C)$ with
	$C$ a $k$-certificate for $(B,\Sigma)$. 
	
	Arguing similarly, no honest party
	will receive a correctly formed $\mathsf{Notify}$ message on anything other than $(B,\Sigma)$. 
	Hence any honest party that outputs grade~1 outputs 
	$(B, \Sigma, C)$ with
	$C$ a $k$-certificate for $(B,\Sigma)$. \qed
\end{proof}

\begin{lemma}\label{lem:prop:commit2}
	Assume fewer than $n/2$ parties are corrupted. If an honest party outputs $(B,\Sigma,C)$ with a nonzero grade in an execution of~$\GC$, then no honest party sends a $\mathsf{Commit}$ message on $(B',\Sigma')\neq(B,\Sigma)$.
\end{lemma}
\begin{proof}
	Say an honest party outputs $(B,\Sigma,C)$ with a nonzero grade. That party must have received a correctly formed $\mathsf{Notify}$ message on $(B, \Sigma)$. Since that $\mathsf{Notify}$ message includes a $k$-certificate $C$ with signatures from more than half the parties, at least one honest party $P$ must have sent a $\mathsf{Commit}$ message on $(B,\Sigma)$.
	This means that $P$ must have received $(B,\Sigma)$ as its output from $\Pi_{\mathsf{Propose}}^{P_\ell}$.
	By Lemma~\ref{lem:prop:cons2}, this means the output of any other honest party from $\Pi_{\mathsf{Propose}}^{P_\ell}$ is either $(B,\Sigma)$ or~$\perp$.
	\qed
\end{proof}

\begin{lemma} \label{lem:gc:cons2}
	Assume fewer than $n/2$ parties are corrupted. If an honest party outputs $(B,\Sigma,C)$ with grade~2 in an execution of~$\GC$, 
	then every honest party outputs a $k$-certificate on $(B, \Sigma)$ with a nonzero grade.
\end{lemma}
\begin{proof} Say an honest party $P$ outputs $(B,\Sigma,C)$ with a grade of~2. 
	By Lemma~\ref{lem:prop:commit2}, this means no honest party sent a
	correctly formed $\mathsf{Commit}$ message on $(B',\Sigma')\neq(B,\Sigma)$; it is thus impossible
	for any honest party to output $(B',\Sigma')\neq(B,\Sigma)$ with a nonzero grade. Since $P$ sends a correctly formed $\mathsf{Notify}$ message on $(B, \Sigma)$ to all honest parties, every honest party will output $(B, \Sigma)$ with a nonzero grade. \qed
\end{proof}

\begin{lemma}\label{lem:gc:live}  
	Assume fewer than $n/2$ parties are corrupted. 	
	Then
	with probability at least~$1/2$ every honest party outputs a $k$-certificate on the same valid $(B, \Sigma)$ with a grade of~2.
\end{lemma}
\begin{proof} 
	The leader $\ell$ chosen in step~2 was honest in step~1 with probability at least~$1/2$. We show that whenever this occurs, every honest party outputs grade~2. Agreement on a valid $(B, \Sigma)$ follows from Lemma~\ref{lem:gc:cons2}.
	
	Assume $\ell$ was honest in step~1.
	Lemma~\ref{lem:prop:honprop} implies that every honest party holds the same valid $(B_\ell, \Sigma_\ell) \neq \perp$ in step~2, and so sends a correctly formed $\mathsf{Commit}$ message on $(B_\ell, \Sigma_\ell)$. Since there are at least $\lceil (n+1)/2 \rceil$ honest parties, the lemma follows. \qed
\end{proof}

In Figure~\ref{prot:bla} we describe our block-agreement protocol~$\BLA$. 

\protocol{Protocol $\BLA$}{A block-agreement protocol $\BLA$.}{prot:bla}{We describe the protocol from the point of view of a party $P$ with input a valid pair~$(B,\Sigma)$.  
	
	\smallskip
	Initialize $(k^*, B^*,\Sigma^*, C^*):=(0, B,\Sigma, \emptyset)$ and $k:=1$. While $k\leq\kappa$ do:
	\begin{enumerate}
		\item At time $(5k-5)\cdot\Delta$, run $\GC$ using input 
		$(k^*, B^*,\Sigma^*, C^*)$ to obtain output 
		$((B,\Sigma, C),g)$.
		\item At time $5k\cdot\Delta$ do:  If $g>0$, set $(k^*, B^*,\Sigma^*, C^*):=(k,B, \Sigma, C)$. If $g=2$, output $(B, \Sigma)$. Increment~$k$. 
	\end{enumerate}
}

\begin{lemma} If $t < n/2$, then $\BLA$ is $t$-secure. 
\end{lemma}
\label{lem:bla:cons}
\begin{proof} Assume fewer than $n/2$ parties are corrupted.
	Let $k$ be the first iteration in which some honest party 
	outputs $(B, \Sigma)$.
	We first show that in every subsequent iteration: (1)~every honest party $P_i$ uses as its input in step~1 a $k_i$-vote on $(B, \Sigma)$; and (2)~corrupted parties cannot construct a $k'$-vote on $(B', \Sigma') \neq (B, \Sigma)$ for any $k' \geq \min_i\{k_i\}$.
	
	Say an honest party outputs $(B, \Sigma)$ in iteration~$k$. Then that party must have output a $k$-certificate for $(B, \Sigma)$ in the execution of $\GC$ in iteration~$k$. 
	By Lemma~\ref{lem:gc:cons2}, this means every honest party output a $k$-certificate on $(B, \Sigma)$ 
	in the same execution of~$\GC$, and so~(1) holds in iteration~$k+1$. Moreover, 
	Lemma~\ref{lem:prop:commit2} implies that no honest party sent a $\mathsf{Commit}$ message on $(B', \Sigma') \neq (B, \Sigma)$ in the execution of~$\GC$, and so~(2) also holds in iteration~\mbox{$k+1$}. Lemma~\ref{lem:gc:val2} implies, inductively, that the stated properties continue to 
	hold in every subsequent iteration. 
	
	It follows from Lemma~\ref{lem:gc:val2} that any other honest party $P$ who generates output in $\BLA$ also outputs~$(B, \Sigma)$, regardless of whether they generate output in iteration~$k$ or a subsequent iteration.
	
	Lemma~\ref{lem:gc:live} shows that in each iteration of~$\BLA$, with probability at least $1/2$ all honest parties output some (the same) valid~$(B, \Sigma)$ in that iteration. Thus, after $\kappa$ iterations all honest parties have generated valid output with probability at least $1-2^{-\kappa}$ (note that all parties terminate after $\kappa$ iterations). \qed
\end{proof}

\end{document}